\documentclass[12pt,letterpaper]{article}
\usepackage{amssymb}
\usepackage{amsmath}
\usepackage{enumerate}
\usepackage{relsize}

\usepackage{amsfonts,amsthm,stmaryrd}

\theoremstyle{plain}

\numberwithin{equation}{section}
\newtheorem{thm}{Theorem}[section]
\newtheorem{cor}[thm]{Corollary}

\newenvironment{exam}
{\begin{flushleft}\textbf{Example}.\enspace}%
{\end{flushleft}}

\newcommand{\complex}{{\mathbb C}}
\newcommand{\real}{{\mathbb R}}

\newcommand{\pscript}{{\mathcal P}}
\newcommand{\rmang}{\mathrm{ang}}
\newcommand{\rmim}{\mathrm{Im}}
\newcommand{\rmre}{\mathrm{Re}}
\newcommand{\rmtr}{\mathrm{tr}}
\newcommand{\ahat}{\widehat{a}}
\newcommand{\xhat}{\widehat{x}}
\newcommand{\yhat}{\widehat{y}}
\newcommand{\omegahat}{\widehat{\omega}}
\newcommand{\atilde}{\widetilde{a}}
\newcommand{\cbar}{\overline{c}}
\newcommand{\abar}{\overline{a}}
\newcommand{\ctimes}{\mathrel{\mathlarger\cdot}}

\newcommand{\ab}[1]{\left|#1\right|}
\newcommand{\brac}[1]{\left\{#1\right\}}
\newcommand{\paren}[1]{\left(#1\right)}
\newcommand{\sqbrac}[1]{\left[#1\right]}
\newcommand{\floors}[1]{\lfloor#1\rfloor}
\newcommand{\sqparen}[1]{{\left[#1\right)}}
\newcommand{\elbows}[1]{{\left\langle#1\right\rangle}}
\newcommand{\ket}[1]{{\left|#1\right>}}
\newcommand{\bra}[1]{{\left<#1\right|}}

\begin{document}

\title{EMERGENCE OF FOUR DIMENSIONS\\ IN THE CAUSAL SET APPROACH\\ TO DISCRETE QUANTUM GRAVITY}
\author{S. Gudder\\ Department of Mathematics\\
University of Denver\\ Denver, Colorado 80208, U.S.A.\\
sgudder@du.edu
}
\date{}
\maketitle

\begin{abstract}
One could begin a study like the present one by simply postulating that our universe is four-dimensional. There are ample reasons for doing this. Experience, observation and experiment all point to the fact that we inhabit a four-dimensional universe. Another approach would be to show that four-dimensions arise naturally from a reasonable model of the universe or multiverse. After reviewing the causal set approach to discrete quantum gravity in Section~1, we shall discuss the emergence of four-dimensions in Section~2. We shall see that certain patterns of four arise that suggest the introduction of a 4-dimensional discrete manifold. In the later sections we shall discuss some consequences of this introduced framework. In particular, we will show that quantum amplitudes can be employed to describe a multiverse dynamics. Moreover, a natural unitary operator together with energy, position and momentum operators will be introduced and their properties studied.
\end{abstract}


\hskip 4pc four dimensions.
\medskip

\section{Causal Set Approach}  

We call a finite poset $(x,<)$ a \textit{causet} and interpret $a<b$ in $x$ to mean that $b$ is in the causal future of $a$
\cite{gud142,hen09,sor03,sur11}. If $x$ and $y$ are causets, we say that $x$ \textit{produces} $y$ (denoted by $x\to y$) if $y$ is obtained from $x$ by adjoining a single maximal element to $x$. If $x\to y$, we call $y$ an \textit{offspring} of $x$. A \textit{labeling} for a causet $x$ of cardinality $\ab{x}$ is a bijection

\begin{equation*}
\ell\colon x\to\brac{1,2,\ldots ,\ab{x}}
\end{equation*}
such that $a,b\in x$ with $a<b$ implies $\ell (a)<\ell (b)$. Two labeled causets $x,y$ are \textit{isomorphic} if there is a \textit{bijection}
$\phi\colon x\to y$ such that $a<b$ in $x$ if and only if $\phi (a)<\phi (b)$ in $y$ and $\ell\sqbrac{\phi (a)}=\ell (a)$ for all $a\in x$. A causet is
\textit{covariant} if it has a unique labeling (up to isomorphisms) and we call a covariant causet a $c$-\textit{causet}
\cite{gud13,gud141,gud142,gud143}. Denote the set of $c$-causets with cardinality $n$ by $\pscript '_n$ and the set of all $c$-causets by
$\pscript '$. It is shown in \cite{gud141} that any $x\in\pscript '$ with $x\ne\emptyset$ has a unique producer in $\pscript '$ and precisely
two offspring in $\pscript '$. It follows that $\ab{\pscript '_n}=2^{n-1}$, $n=1,2,\ldots\,$. Two elements $a,b\in x$ are \textit{comparable} if $a<b$ or $b<a$. The \textit{height} $h(a)$ of $a\in x$ is the cardinality, minus one, of the longest path in $x$ that ends with $a$. It is shown in \cite{gud141} that a causet $x$ is covariant if and only if $a,b\in x$ are comparable whenever $h(a)\ne (b)$.

If $x\in\pscript '$ we call the sets

\begin{equation*}
S_j(x)=\brac{a\in x\colon h(a)=j},\quad j=0,1,2,\ldots
\end{equation*}
\textit{shells} and the sequence of integers $s_i(x)=\ab{S_j(x)}$, $j=0,1,\ldots$, is the \textit{shell sequence}. A $c$-causet is uniquely determined by its shell sequence and we think of $\brac{s_i(x)}$ as describing the ``shape'' or geometry of $x$. The tree $(\pscript ',\shortrightarrow )$ can be thought of as a growth model and an $x\in\pscript '_n$ is a possible universe at step (time) $n$. An instantaneous universe $x\in\pscript '_n$ grows one element at a time in one of two ways. If $x$ has shell sequence $\paren{s_o(x),s_i(x),\ldots ,s_m(x)}$ then $x\to x_0$ or $x\to x_1$ where
$x_0,x_1\in\pscript '_{n+1}$ have shell sequences

\begin{align*}
\paren{s_0(x),s_1(x),\ldots ,s_m(x)+1}\\
\intertext{and}
\paren{s_0(x),s_1(x),\ldots ,s_m(x),1}
\end{align*}
respectively. In this way, we recursively order the $c$-causets in $\pscript '_n$ using the notation $x_{n,j}$, $n=1,2,\ldots$, $j=0,1,\ldots,2^{n-1}-1$ where $n=\ab{x_{n,j}}$. For example, in terms of their shell sequences we have:

\begin{align*}
x_{1,0}&=(1),x_{2,0}=(2),x_{2,1}=(1,1),x_{3,0}=(3),x_{3,1}=(2,1),x_{3,2}=(1,2),x_{3,3}=(1,1,1)\\
x_{4,0}&=(4),x_{4,1}=(3,1),x_{4,2}=(2,2),x_{4,3}=(2,1,1),x_{4,4}=(1,3),x_{4,5}=(1,2,1)\\
x_{4,6}&=(1,1,2),x_{4,7}=(1,1,1,1)
\end{align*}

In this model, we view a $c$-causet as a framework or scaffolding for a possible universe. The vertices of $x$ represent small cells that can be empty or occupied by a particle. The shell sequence for $x$ gives the geometry of the framework \cite{gud142}. Notice that this is a multiverse model in which infiite paths in $(\pscript ',\shortrightarrow )$ represent the histories of ``completed'' evolved universes.

We now describe the evolution of a universe as a quantum sequential growth process. In such a process, the probabilities of competing geometries are determined by quantum amplitudes. These amplitudes provide interferences that are characteristic of quantum systems. A
\textit{transition amplitude} is a map $\atilde\colon\pscript '\times\pscript '\to\complex$ satisfying $\atilde (x,y)=0$ if $x\not\to y$ and
$\sum _{y\in\pscript '}\atilde (x,y)=1$. Since $x_{n,j}$ only has the offspring $x_{n+1,2j}$ and $x_{n+1,2j+1}$ we have that

\begin{equation*}
\sum _{k=0}^1\atilde (x_{n,j},x_{n+1,2j+k})=1
\end{equation*}
for all $n=1,2,\ldots$, $j=0,1,\ldots ,2^{n-1}-1$. We call $\atilde$ a \textit{unitary transition amplitude} (uta) if $\atilde$ also satisfies
$\sum _{y\in\pscript '}\ab{\atilde (x,y)}^2=1$. One might suspect that the restrictions on a uta are so strong that the possibilities are very limited. This would be true if $\atilde$ were real valued. In this case, $\atilde (x,y)$ is 0 or 1. However, in the complex case, the next result, which is proved in \cite{gud143}, shows that there are a continuum of possibilities.

\begin{thm}       
\label{thm11}
Two compex numbers $a,b$ satisfy $a+b=\ab{a}^2+\ab{b}^2=1$ if and only if there exists a $\theta\in\sqparen{0,\pi}$ such that
$a=\cos\theta e^{i\theta}$, $b=-i\sin\theta e^{i\theta}$. Moreover, $\theta$ is unique.
\end{thm}

If $\atilde\colon\pscript '\times\pscript '\to\complex$ is a uta, we call 

\begin{equation*}
c_{n,j}^k=\atilde(x_{n,j},x_{n+1,2j+k}),\quad k=0,1
\end{equation*}
the 2-\textit{dimensional coupling constants} for $\atilde$. It follows from Theorem~\ref{thm11} that there exist $\theta _{n,j}\in\sqparen{0,\pi}$ such that

\begin{equation*}
c_{n,j}^0=\cos\theta _{n,j}e^{i\theta _{n,j}},c_{n,j}^1=-i\sin\theta _{n,j}e^{i\theta _{n,j}}
\end{equation*}
Hence, $c_{n,j}^0+c_{n,j}^1=\ab{c_{n,}^0}^2+\ab{c_{n,j}^1}^2=1$ for all $n=1,2,\ldots$, $j=0,1,\ldots ,2^{n-1}-1$.

A \textit{path} in $\pscript '$ is a sequence $\omega =\omega _1\omega _2\cdots$ where $\omega _i\in\pscript '_i$ and
$\omega _i\to\omega _{i+1}$. Similarly, an $n$-\textit{path} has the form $\omega =\omega _1\omega _2\cdots\omega _n$ where again
$\omega _i\to\omega _{i+1}$. We denote the set of paths by $\Omega '$ and the set of $n$-paths by $\Omega '_n$. Since every $x\in\pscript '_n$ has a unique $n$-path terminating at $x$, we can identify $\pscript '_n$ with $\Omega '_n$ and write $\pscript '_n\approx\Omega '_n$. If $\atilde$ is a uta and $\omega =\omega _1\omega _2\cdots\omega _n\in\Omega '_n$, we define the \textit{amplitude} of $\omega$ to be

\begin{equation*}
a(\omega )=\atilde (\omega _1,\omega _2)\atilde (\omega _2,\omega _3)\cdots\atilde (\omega _{n-1},\omega _n)
\end{equation*}
Moreover, we define the \textit{amplitude} of $x\in\pscript '_n$ to be $a(\omega )$ where $\omega$ is the unique path in $\Omega '_n$ that terminates at $x$. For $A\subseteq\pscript '_n$ we define the \textit{amplitude} of $A$ to be

\begin{equation*}
a(A)=\sum\brac{a(x)\colon x\in A}
\end{equation*}
and the $q$-\textit{measure} of $A$ to be $\mu _n(A)=\ab{a(A)}^2$. We conclude that $\mu _n\colon 2^{\pscript '_n}\to\real ^+$ and it is not hard to verify that $\mu (\pscript '_n)=1$ \cite{hen09,sor94,sor03}.

In general, $\mu _n$ is not additive so it is not a measure. For this reason $\mu _n(A)$ is interpreted as the quantum propensity of $A$ instead of the quantum probability. Although $\mu _n$ is not additive, it satisfies the \textit{grade}-2 \textit{additivity condition} \cite{hen09,sor94,sor03}: if
$A,B,C\in 2^{\pscript '_n}$ are mutually disjoint then

\begin{equation*}
\mu _n(A\cup B\cup C)=\mu _n(A\cup B)+\mu _n(A\cup C)+\mu _n(B\cup C)-\mu _n(A)-\mu _n(B)-\mu _n(C)
\end{equation*}

Because of the lack of additivity we have, in general, that

\begin{equation}         
\label{eq11}
\mu _n\paren{\brac{x,y}}\ne\mu _n(x)+\mu _n(y)
\end{equation}
for $x,y\in\pscript '_n$. If \eqref{eq11} holds, we say that $x$ and $y$ \textit{interfere} and otherwise we say that $x$ and $y$ \textit{do not interfere}. It is shown in \cite{gud143} that if $x$ and $y$ have the same producer, then $x$ and $y$ do not interfere.

\section{Patterns of Four} 

This section shows that various patterns of four occur in $\pscript '_n$. The \textit{height} $h(x)$ of $x\in\pscript '$ is the cardinality minus one of the longest paths in $x$. Equivalently, $h(x)$ is the number of shells minus one in $x$. For example, since $x_{4,5}=(1,2,1)$ we have $h(x_{4,5})=2$ The
\textit{height sequence} of $\pscript '_n$ is the sequence of integers

\begin{equation*}
\paren{h(x_{n,0}),h(x_{n,1}),\ldots ,h(x_{n,2^{n-1}-1})}
\end{equation*}
We now display the height sequences of $\pscript '_n$ for the first few values of $n$.

\begin{align*}
\pscript '_1\colon&(0),\pscript '_2\colon (0,1),\pscript '_3\colon (0,1,1,2),\pscript '_4\colon (0,1,1,2,1,2,2,3)\\
\pscript '_5\colon&(0,1,1,2,1,2,2,3,1,2,2,3,2,3,3,4)\\
\pscript '_6\colon&(0,1,1,2,1,2,2,3,1,2,2,3,2,3,3,4,1,2,2,3,2,3,3,4,2,3,3,4,3,4,4,5)
\end{align*}

Notice the compelling patterns of four that stand out. For example, in $\pscript '_6$ we have $(0,1,1,2)$, $(1,2,2,3)$, $(1,2,2,3)$, $(2,3,3,4)$,
$(1,2,2,3)$, $(2,3,3,4)$, $(2,3,3,4)$, $(3,4,4,5)$. These patterns have the form $r,r+1,r+1,r+2$. One might also say that there are patterns of two and patterns of eight, but these are not as compelling. We now show why this four pattern occurs. If $x\to y$ then $h(y)=h(x)$ or $h(x)+1$. Hence, if
$\pscript '_n$ has height sequence,

\begin{equation*}
(r_1,r_2,\ldots ,r_{2^{n-1}-1})
\end{equation*}
then $\pscript '_{n+1}$ has height sequence

\begin{equation*}
(r_1,r_1+1,r_2,r_2+1,\ldots ,r_{2^{n-1}-1},r_{2^{n-1}-1}+1)
\end{equation*}
Applying this reasoning again shows that $\pscript '_{n+2}$ has height sequence

\begin{align*}
(r_1,r_1+1,r_1+1,r_1+2,&r_2,r_2+1,r_2+1,r_2+2,\ldots ,\\
&r_{2^{n-1}-1},r_{2^{n-1}-1}+1,r_{2^{n-1}-1}+1,r_{2^{n-1}-1}+2)
\end{align*}

As an aside, it is of interest to consider the number of causets $\tau _n(r)$ in $\pscript '_n$ with height $r=0,1,\ldots ,n-1$.
For example, $\tau _5(0)=1$, $\tau _5(1)=4$, $\tau _5(2)=6$, $\tau _5(3)=4$, $\tau _5(4)=1$. This suggests the following result.

\begin{thm}       
\label{thm21}
For $n=1,2,\ldots $, $r=0,1,\ldots ,n-1$ we have that $\tau _n(r)=\binom{n-1}{r}$.
\end{thm}
\begin{proof}
We employ induction on $n$. The result $\tau _1(0)=1$ is clearly true. Assume the result holds for $\pscript '_n$ and consider causets in
$\pscript '_{n+1}$ with height $r$. Now any $x\in\pscript '_n$ with $h(x)=r$ produces a $y\in\pscript '_n$ with $h(y)=r$. Also any $x\in\pscript '_n$
with $h(x)=r-1$ (we can assume $r\ne 0$) produces a $y\in\pscript '_n$ with $h(y)=r$. By the induction hypothesis

\begin{equation*}
\tau _{n+1}(r)=\tau _n(r)+\tau _n(r-1)=\binom{n-1}{r}+\binom{n-1}{r-1}
\end{equation*}
But a well-known combinatorial identity gives

\begin{equation*}
\binom{n-1}{r}+\binom{n-1}{r-1}=\binom{n}{r}\qedhere
\end{equation*}
\end{proof}

\begin{cor}       
\label{cor22}
The maximum of $\tau _n(r)$ occurs when $r=\floors{\frac{n-1}{2}}$ where $\floors{\cdot}$ is the floor function.
\end{cor}

We now consider a second example of four patterns. If $x\in\pscript '$ has shell sequence $(s_0(x),s_1(x),\ldots ,s_m(x))$, $s_m(x)>0$ we define
$w(x)=s_m(x)$. Thus, $w(x)$ is the cardinality of the highest shell of $x$. The \textit{width sequence} of $\pscript '_n$ is the sequence of integers

\begin{equation*}
\paren{w(x_{n,0}),w(x_{n,1}),\ldots ,w(x_{n,2^{n-1}-1})}
\end{equation*}
We next display the width sequences of $\pscript '_n$ for the first few values of $n$.

\begin{align*}
\pscript '_1\colon&(1),\pscript '_2\colon (2,1),\pscript '_3\colon (3,1,2,1),\pscript '_4\colon (4,1,2,1,3,1,2,1)\\
\pscript '_5\colon&(5,1,2,1,3,1,2,1,4,1,2,1,3,1,2,1)\\
\pscript '_6\colon&(6,1,2,1,3,1,2,1,4,1,2,1,3,1,2,1,5,1,2,1,3,1,2,1,4,1,2,1,3,1,2,1)
\end{align*}
As before, the patterns of four $(r,1,2,1)$ stand out clearly.

Our last example of a four pattern stems from quantum mechanics. What we say here applies to any $\pscript '_n$ for $n\ge 3$ but for simplicity let us consider $\pscript '_3=\brac{x_{3,0},x_{3,1},x_{3,2},x_{3,3}}$. Since $x_{3,0}$ and $x_{3,1}$ have the same producer, as mentioned in Section~1,
$x_{3,0}$ and $x_{3,1}$ do not interfere. Thus, the pair $(x_{3,0},x_{3,1})$ act classically with no apparent quantum effects. The same applies for the pair $(x_{3,2},x_{3,3})$. Calling such pairs \textit{siblings}, if we want to include siblings and quantum effects the smallest applicable set is the 4-tuple $(x_{3,0},x_{3,1},x_{3,2},x_{3,3})$. In general, the 4-tuple would have the form

\begin{equation*}
\paren{x_{n,4j},x_{n,4j+1},x_{n,4j+2},x_{n,4j+3}}
\end{equation*}
There are precisely \textit{four} interfering pairs in such a 4-tuple. These are $(x_{n,4j},x_{n,4j+2})$, $(x_{n,4j},x_{n,4j+3})$,
$(x_{n,4j+1},x_{n,4j+2})$ and $(x_{n,4j+1},x_{n,4j+3})$.

\begin{exam} 
We now show that, in general, $x_{3,0}$ and $x_{3,2}$ interfere. In terms of the coupling constants we have that
\begin{equation}         
\label{eq21}
\mu _3\paren{\brac{x_{3,0},x_{3,2}}}=\ab{a(x_{3,0})+a(x_{3,2})}^2=\ab{c_{1,0}^0c_{2,0}^0+c_{1,0}^1c_{2,1}^0}^2
\end{equation}
On the other hand

\begin{equation}         
\label{eq22}
\mu _3(x_{3,0})+\mu _3(x_{3,2})=\ab{a(x_{3,0})}^2+\ab{a(x_{3,2})}^2=\ab{c_{1,0}^0}^2\ab{c_{2,0}^0}^2+\ab{c_{1,0}^1}^2\ab{c_{2,1}^0}^2
\end{equation}
But \eqref{eq21} and \eqref{eq22} do not agree unless

\begin{equation*}
\rmre (c_{1,0}^0c_{2,0}^0\cbar _{1,0}^1\cbar _{2,1}^0)=0
\end{equation*}
so $x_{3,0}$ and $x_{3,2}$ interfere, in general. The same reasoning also shows that $(x_{3,0},x_{3,3})$, $(x_{3,1},x_{3,2})$, $(x_{3,1},x_{3,3})$ are interfering pairs, in general.
\end{exam}

This argument also shows that $x_{3,0}$ and $x_{3,1}$ do not interfere. In this case \eqref{eq21} becomes

\begin{equation*}
\mu _3\paren{\brac{x_{3,0},x_{3,1}}}=\ab{c_{1,0}^0c_{2,0}^0+c_{1,0}^0c_{2,0}^1}^2
  =\ab{c_{1,0}^0}^2\ab{c_{2,0}^0+c_{2,0}^1}^2=\ab{c_{1,0}^0}^2
\end{equation*}
and \eqref{eq22} becomes

\begin{equation*}
\mu _3(x_{3,0})+\mu _3(x_{3,1})=\ab{c_{1,0}^0}^2\paren{\ab{c_{2,0}^0}^2+\ab{c_{2,0}^1}^2}=\ab{c_{1,0}^0}^2
\end{equation*}
Hence, \eqref{eq21} and \eqref{eq22} agree.

\section{Four Dimensional Discrete Manifold} 

The previous section suggests that important patterns occur for $c$-causet 4-tuples of the form $(x_{n,4j},x_{n,4j+1},x_{n,4j+2},x_{n,4j+3})$. We can write such 4-tuples as $x_{n,4j+k}$, $k=0,1,2,3$. This indicates that instead of considering the full set of $c$-causets $\pscript '$ we should concentrate on the set

\begin{equation*}
\pscript =\cup\brac{\pscript '_n\colon n\hbox{ odd}}
\end{equation*}
of odd cardinality $c$-causets. We then define $\pscript _n=\pscript '_{2n-1}$ to be the collection of $c$-causets with cardinality $2n-1$,
$n=1,2,\ldots$, so that $\pscript =\cup\pscript _n$. We now have

\begin{equation*}
\ab{\pscript _n}=\ab{\pscript '_{2n-1}}=2^{2n-2}=4^{n-1}
\end{equation*}
and as before we order the $c$-causets in $\pscript _n$ as

\begin{equation*}
\pscript _n=\brac{x_{n,0},x_{n,1},\ldots ,x_{n,4^{n-1}-1}}
\end{equation*}
In this case, each $x\in\pscript$ except $x_{1,0}$ has a unique producer and each $x\in\pscript$ has four offspring. In particular
$x_{n,j}\to x_{n+1,4j+k}$, $k=0,1,2,3$, so $(\pscript ,\shortrightarrow )$ becomes a tree that we interpret as a sequential growth process. The main difference is that $x\to y$ if $y$ is obtained by first adjoining a maximal element $a$ to $x$ and then adjoining a second maximal element $b$ to
$x\cup\brac{a}$ so that $y=x\cup\brac{a,b}$.

In this framework, $\pscript$ has the structure of a discrete 4-manifold. A \textit{tangent vector} at $x\in\pscript$ is a pair $(x,y)$ where
$x\in\pscript _n$, $y\in\pscript _{n+1}$ and $x\to y$. Since every $x\in\pscript$ has four offspring, there are four tagent vectors at $x$. We denote the tangent vectors at $x_{n,j}$ by $d_{n,j}^k$, $k=0,1,2,3$ where

\begin{equation*}
d_{n,j}^k=(x_{n,j},x_{n+1,4j+k})
\end{equation*}
We say that two tangent vectors are \textit{incident} if they have the forms $(x,y),(y,z)$. As before, an $n$-\textit{path} in $\pscript$ is a sequence
$\omega =\omega _1\omega _2\cdots\omega _n$ where $\omega _i\in\pscript _i$ and $\omega _i\to\omega _{i+1}$. We denote the set of $n$-paths by $\Omega _n$. We can consider an $n$-path as a sequence of tangent vectors 

\begin{equation*}
\omega =d_{1,0}^{k_1}d_{2,j_2}^{k_2}\cdots d_{n-1,j_{n-1}}^{k_{n-1}}
\end{equation*}
where each tangent vector is incident to the next.

A \textit{transition amplitude} $\atilde\colon\pscript\times\pscript\to\complex$ and a \textit{unitary transition amplitude} (uta) are defined as before. Moreover, we call

\begin{equation*}
c_{n,j}^k=\atilde (x_{n,j},x_{n+1,4j+k}),\quad k=0,1,2,3
\end{equation*}
the \textit{coupling constants} for $\atilde$. If $\atilde$ is a uta, we have

\begin{equation}         
\label{eq31}
\sum _{k=0}^3c_{n,j}^k=\sum _{k=0}^3\ab{c_{n,j}^k}=1
\end{equation}
for $n=1,2,\ldots$, $j=0,1,\ldots ,4^{n-1}-1$. As in Section~2, we can identify $\pscript _n$ with $\Omega _n$ and write
$\pscript _n\approx\Omega _n$. If $\atilde$ is a uta and $\omega =\omega _1\omega _2\cdots\omega _n\in\Omega _n$ we define the \textit{amplitude} of $\omega$ to be

\begin{equation*}
a(\omega )=\atilde (\omega _1,\omega _2)\atilde (\omega _2,\omega _3)\cdots\atilde (\omega _{n-1},\omega _n)
\end{equation*}
Also, we define the \textit{amplitude} of $x\in\pscript$ to be $a(\omega )$ where $\omega$ is the unique path in $\Omega _n$ that terminates at $x$.

Let $H_n$ be the Hilbert space

\begin{equation*}
H_n=L_2(\Omega _n)=L_2(\pscript _n)=\brac{f\colon\pscript _n\to\complex}
\end{equation*}
with the standard inner product

\begin{equation*}
\elbows{f,g}=\sum _{x\in\pscript _n}\overline{f(x)}g(x)
\end{equation*}
Let $\xhat _{n,j}$ be the unit vector in $H_n$ given by the characteristic function $\chi _{x_{n,j}}$. Then $\dim H_n=4^{n-1}$ and
$\brac{\xhat _{n,}\colon j=0,1,\ldots ,4^{n-1}-1}$ forms an orthonormal basis for $H_n$. For the remainder of this section $\atilde$ is a uta with corresponding coupling constants $c_{n,j}^k$. We now describe the quantum dynamics generated by $\atilde$. Define the operators
$U_n\colon H_n\to H_{n+1}$ by

\begin{equation*}
U_n\xhat _{n,j}=\sum _{k=0}^3c_{n,j}^k\xhat _{n+1,4j+k}
\end{equation*}
and extend $U_n$ to $H_n$ by linearity. The next few theorems generalize results in \cite{gud143}.

\begin{thm}       
\label{thm31}
{\rm (i)}\enspace The adjoint of $U_n$ is given by $U_n^*\colon H_{n+1}\to H_n$, where

\begin{equation}         
\label{eq32}
U_n^*\xhat _{n+1,4j+k}=\cbar _{n,j}^k\xhat _{n,j},\quad k=0,1,2,3
\end{equation}
{\rm (ii)}\enspace $U_n$ is a partial isometry with $U_n^*U_n=I_n$ and

\begin{equation}         
\label{eq33}
U_nU_n^*=\sum _{j=0}^{4^{n-1}-1}\ket{\sum _{k=0}^3c_{n,j}^k\xhat _{n+1,4j+k}}\bra{\sum _{k=0}^3c_{n,j}^k\xhat _{n+1,4j+k}}
\end{equation}
\end{thm}
\begin{proof}
(i)\enspace To show that \eqref{eq32} holds we have that

\begin{align*}
\elbows{U_n^*\xhat _{n+1,4j'+k'},\xhat _{n,j}}&=\elbows{\xhat _{n+1,4j'+k'},U_nx_{n,j}}\\
  &=\elbows{\xhat _{n+1,4j'+k'},\sum _{k=0}^3c_{n,j}^k\xhat _{n+1,4j+k}}\\
  &=c_{n,j'}^{k'},\delta _{j,j'}=\elbows{\cbar _{n,j'}^{k'}\xhat _{n,j'},\xhat _{n,j}}
\end{align*}
(ii)\enspace To show that \eqref{eq32} holds we have that

\begin{equation*}
U_n^*U_n\xhat _{n,j}=\sum _{k=0}^3c_{n,j}^kU_n^*\xhat _{n+1,4j+k}=\sum _{k=0}^3\ab{c_{n,j}^k}^2\xhat _{n,j}=\xhat _{n,j}
\end{equation*}
Since $\brac{\xhat _{n,j}\colon j=0,1,\ldots ,4^{n-1}-1}$ forms an orthonormal basis for $H_n$, the result follows. Equation \eqref{eq33} holds because it is well-known that $U_nU_n^*$ is the projection onto the range of $U_n$.
\end{proof}

It follows from Theorem~\ref{thm31} that the dynamics $U_n\colon H_n\to H_{n+1}$ for a uta $\atilde$ is an isometric operator. As usual a \textit{state} on $H_n$ is a positive operator $\rho$ on $H_n$ with $\rmtr (\rho )=1$. A \textit{stochastic state} on $H_n$ is a state $\rho$ that satisfies
 $\elbows{\rho 1_n,1_n}=1$ where $1_n=\chi _{\pscript _n}$; that is $1_n(x)=1$ for all $x\in\pscript _n$. Notice that $U_n^*1_{n+1}=1_n$.

\begin{cor}       
\label{cor32}
{\rm (i)}\enspace If $\rho$ is a state on $H_n$, then $U_n\rho U_n^*$ is a state on $H_{n+1}$
{\rm (ii)}\enspace If $\rho$ is a stochastic state on $H_n$, then $U_n\rho U_n^*$ is a stochastic state on $H_{n+1}$.
\end{cor}

Corollary~\ref{cor32} shows that $\rho\to U_n\rho U_n^*$ gives a quantum dynamics for states. We now show that a natural stochastic state is generated by a uta $\atilde$. Since

\begin{equation*}
\elbows{\xhat _{n+1,4j+k},U_n\xhat _{n,j}}=c_{n,j}^k=\atilde (x_{n,j},x_{n+1,4j+k})
\end{equation*}
we have for any $\omega =\omega _1\omega _2\cdots\omega _n\in\Omega _n$ that

\begin{equation*}
a(\omega )
  =\elbows{\omegahat _2,U_1\omegahat _1}\elbows{\omegahat _3,U_2\omegahat _2}\cdots\elbows{\omegahat _n,U_{n-1}\omegahat _{n-1}}
\end{equation*}
For $\omega\in\Omega _n$, let $\omegahat =\chi _{\brac{\omega}}\in H_n$. Define the operator $\rho _n$ on $H_n$ by
$\elbows{\omegahat ,\rho _n\omegahat '}=a(\omega )\overline{a(\omega ')}$. Equivalently, we have that
$\elbows{\xhat ,\rho _n\yhat}=a(x)\overline{a(y)}$ for every $x,y\in\pscript _n$. A straightforward generalization of Theorem~2.4 \cite{gud143} shows that $\rho _n$ is a stochastic state on $H_n$. We call $\rho _n$ the \textit{amplitude state} corresponding to $\atilde$.

We have seen that $x_{n,j}\in\pscript _n$ produces four offspring $x_{n+1,4j+k}\in\pscript _{n+1}$, $k=0,1,2,3$. We call the set

\begin{equation*}
(x_{n,j}\shortrightarrow )=\brac{x_{n+1,4j+k}\colon k=0,1,2,3}\subseteq\pscript _{n+1}
\end{equation*}
the \textit{one-step causal future} of $x_{n,j}$. For simplicity we write $(x\shortrightarrow )$ for the one-step causal future of $x\in\pscript _n$ and we use the notation $(x\shortrightarrow )^\wedge =\chi _{(x\shortrightarrow )}$. A straightforward generalization of Theorem~2.5 \cite{gud143} shows that the amplitude state sequence $\rho _n$ is \textit{consistent} in the sense that

\begin{equation*}
\elbows{(x\shortrightarrow )^\wedge ,\rho _{n+1}(y\shortrightarrow )^\wedge}=\elbows{\xhat ,\rho _n\yhat}
\end{equation*}
for every $x,y\in\pscript _n$. Consistency is important because it states that the probabilities and propensities given by the dynamics $\rho _n$ are conserved in time.

A vector $v\in H_n$ is a \textit{stochastic state vector} if $\|v\|=\elbows{v,1_n}=1$. We call the vector

\begin{equation}         
\label{eq34}
\ahat _n=\sum _{j=0}^{4^{n-1}-1}a(x_{n,j})\xhat _{n,j}\in H_n
\end{equation}
an \textit{amplitude vector}. Of course, $\ahat _n$ is a stochastic state vector \cite{gud143}. It is easy to check that
$\rho _n=\ket{\ahat _n}\bra{\ahat _n}$ so $\rho _n$ is a pure state. The next result shows that $\ahat _n$ has the expected properties.

\begin{thm}       
\label{thm33}
{\rm (i)}\enspace If $v\in H_n$ is a stochastic state vector then so is $U_nv\in H_{n+1}$.
{\rm (ii)}\enspace $U_n\ahat _n=\ahat _{n+1}$.
{\rm (iii)}\enspace $U_n^*\ahat _{n+1}=\ahat _n$.
\end{thm}
\begin{proof}
(i)\enspace This follows from the fact that $U_n$ is isometric and $U_n^*1_{n+1}=1_n$.
(ii)\enspace This follows from

\begin{align*}
U_n\ahat _n&=U_n\sum _{j=0}^{4^{n-1}-1}a(x_{n,j})\xhat _{n,j}=\sum _{j=0}^{4^{n-1}-1}a(x_{n,j})\sum _{k=0}^3c_{n,j}^k\xhat _{n+1,4j+k}\\
  &=\sum _{j=0}^{4^{n-1}-1}\sum _{k=0}^3a(x_{n,j})c_{n,j}^k\xhat _{n+1,4j+k}\\
  &=\sum _{j=0}^{4^{n-1}-1}\sum _{k=0}^3a(x_{n+1,4j+k})\xhat _{n+1,4j+k}=\ahat _{n+1}
\end{align*}
(iii)\enspace Applying (ii), we have that

\begin{equation*}
U_n^*\ahat _{n+1}=U_n^*U_n\ahat _n=\ahat _n\qedhere
\end{equation*}
\end{proof}

\section{Coupling Constant Symmetry} 

In Section~3 we considered 4-dimensional coupling constants $c_{n,j}^k=a(x_{n,4j+k})$, $n=1,2,\ldots$, $j=0,1,\ldots ,4^{n-1}-1$, $k=0,1,2,3$ that satisfied \eqref{eq31}. How do we know that such coupling constants exist? Certainly there exist trivial coupling constants satisfying
$c_{n,j}^k=0\hbox{ or }1$, but what about nontrivial coupling constants? One way to construct such coupling constants is by employing tensor products. This construction does not give all possible coupling constants but it does produce a large number of them and it might have physical significance.

Coupling constants can be thought of as stochastic unit vectors in $\complex ^4$ of the form

\begin{equation*}
u=(c_{n,j}^0,c_{n,j}^1,c_{n,j}^2,c_{n,j}^3)
\end{equation*}
Such vectors satisfy

\begin{equation*}
\elbows{u,1_4}=\|u\|=1
\end{equation*}
where $1_4$ is the vector $(1,1,1,1)$. Now let $v_1$ and $v_2$ be stochastic unit vectors in $\complex ^2$. (These are particular types of qubit states.) Theorem~\ref{thm11} shows that there are a continuum of such vectors. The tensor product $v_1\otimes v_2$ is a stochastic unit vector in
$\complex ^4$. Indeed, since $1_4=1_2\otimes 1_2$ we have

\begin{align*}
\elbows{v_1\otimes v_2,1_4}&=\elbows{v_1\otimes v_2,1_2\otimes 1_2}=\elbows{v_1,1_2}\elbows{v_2,1_2}=1\\
\intertext{also}
\|v_1\otimes v_2\|&=\|v_1\|\,\|v_2\|=1
\end{align*}
Now suppose that $e_{n,j}^k$, $f_{n,j}^k$, $k=0,1$, are 2-dimensional coupling constants as studied in Section~1. Employing the notation
$e_{n,j}^k=(e_{n,j}^0,e_{n,j}^1)$, $f_{n,j}^k=(f_{n,j}^0,f_{n,j}^1)$, let $c_{n,j}^k=e_{n,j}^k\otimes f_{n,j}^k$. We conclude that $c_{nj}^k$ are 4-dimensional coupling constants that we call \textit{product coupling constants}. To be explicit, we have

\begin{equation*}
c_{n,j}^k=\begin{bmatrix}\noalign{\medskip}
\begin{bmatrix}e_{n,j}^0\\ e_{n,j}^1\end{bmatrix}&f_{n,j}^0\\\noalign{\medskip}
  \begin{bmatrix}e_{n,j}^0\\ e_{n,j}^1\end{bmatrix}&f_{n,j}^1\\\noalign{\medskip}\end{bmatrix}
  =\begin{bmatrix}\noalign{\smallskip}
  e_{n,j}^0f_{n,j}^0\\\noalign{\smallskip}e_{n,j}^1f_{n,j}^0\\\noalign{\smallskip}
  e_{n,j}^0f_{n,j}^1\\\noalign{\smallskip}e_{n,j}^1f_{n,j}^1\\\noalign{\smallskip}\end{bmatrix}
\end{equation*}
Hence, $c_{n,j}^0=e_{n,j}^0f_{n,j}^0$, $c_{n,j}^1=e_{n,j}^1f_{n,j}^0$, $c_{n,j}^2=e_{n,j}^0f_{n,j}^1$, $c_{n,j}^3=e_{n,j}^1f_{n,j}^1$. This section shows that product coupling constants have particularly interesting properties.

Let $a,b$ be complex numbers satisfying $a+b=\ab{a}^2+\ab{b}^2=1$. We have seen in Theorem~\ref{thm11} that there exists a unique
$\theta\in\sqparen{0,\pi}$ such that $a=\cos\theta e^{i\theta}$. We call $\theta$ the \textit{angle for} $a$ and write $\theta =\rmang (a)$. We now ``double-down'' the pair $(a,b)$ to form the unitary matrix.

\begin{equation*}
A=\begin{bmatrix}a&b\\ b&a\end{bmatrix}
\end{equation*}
The matrix $A$ is also \textit{stochastic} in the sense that $A1_2=1_2$. It follows that $u_1=2^{-1/2}1_2$ is a unit eigenvector of $A$ with corresponding eigenvalues $1$. It is easy to check that

\begin{equation*}
u_2=2^{-1/2}=\begin{bmatrix}1\\-1\end{bmatrix}
\end{equation*}
is the other unit eigenvector of $A$ with corresponding eigenvalue

\begin{equation*}
a-b=2a-1=e^{i2\theta}
\end{equation*}
Now let $c+d=\ab{c}^2+\ab{d}^2=1$ and again form the stochastic unitary matrix

\begin{equation*}
B=\begin{bmatrix}c&d\\ d&c\end{bmatrix}
\end{equation*}
Letting $\phi=\rmang (c)$ we have the following result.

\begin{thm}       
\label{thm41}
$A\otimes B$ is a $4\times 4$ stochastic unitary matrix with eigenvalues $1$, $e^{i2\theta}$, $e^{i2\phi}$, $e^{i2(\theta +\phi )}$ and corresponding unit eigenvectors $u_1\otimes u_1$, $u_2\otimes u_1$, $u_1\otimes u_2$, $u_2\otimes u_2$.
\end{thm}
\begin{proof}
$A\otimes B$ is unitary because

\begin{equation*}
(A\otimes B)(A\otimes B)^*=(A\times B)(A^*\otimes B^*)=(AA^*\otimes BB^*)=I_2\otimes I_2=I_4
\end{equation*}
Also, $A\otimes B$ is stochastic because

\begin{equation*}
(A\otimes B)1_4=(A\otimes B)(1_2\otimes 1_2)=A1_2\otimes B1_2=1_2\otimes 1_2=1_4
\end{equation*}
We have seen that $A\otimes B(u_1\otimes u_1)=u_1\otimes u_1$. Moreover,

\begin{align*}
A\otimes B(u_2\otimes u_1)&=Au_2\otimes Bu_1=e^{i2\theta}u_2\otimes u_1\\
A\otimes B(u_1\otimes u_2)&=Au_1\otimes Bu_2=e^{i2\phi}u_1\otimes u_2\\
A\otimes B(u_2\otimes u_2)&=Au_2\otimes Bu_2=e^{i2(\theta +\phi)}u_2\otimes u_2\qedhere
\end{align*}
\end{proof}

Let $\atilde$ be a uta on $\pscript$ with product coupling constants $c_{n,j}^k$, $n=1,2,\ldots$, $j=0,1,\ldots ,4^{n-1}-1$, $k=0,1,2,3$. Since
$c_{n,j}^k$ are product coupling constants, for fixed $n,j$ we have that $c_{n,j}^k=e_{n,j}^k\otimes f_{n,j}^k$ for 2-dimensional coupling constants
$e_{n,j}^k$ and $f_{n,j}^k$. We have seen that

\begin{equation*}
A_j=\begin{bmatrix}e_{n,j}^0&e_{n,j}^1\\\noalign{\smallskip}e_{n,j}^1&e_{n,j}^0\end{bmatrix}\,,\quad
B_j=\begin{bmatrix}f_{n,j}^0&f_{n,j}^1\\\noalign{\smallskip}f_{n,j}^1&f_{n,j}^0\end{bmatrix}
\end{equation*}
are stochastic unitary matrices so it follows from Theorem~\ref{thm41} that $C_j=A_j\otimes B_j$ is also stochastic and unitary. From the definition of the tensor product we have that

\begin{equation}        
\label{eq41}
C_j=\begin{bmatrix}c_{n,j}^0&c_{n,j}^1&c_{n,j}^2&c_{n,j}^3\\\noalign{\smallskip}
  c_{n,j}^1&c_{n,j}^0&c_{n,j}^3&c_{n,j}^2\\\noalign{\smallskip}c_{n,j}^2&c_{n,j}^3&c_{n,j}^0&c_{n,j}^1\\\noalign{\smallskip}
  c_{n,j}^3&c_{n,j}^2&c_{n,j}^1&c_{n,j}^0\end{bmatrix}
\end{equation}
applying Theorem~\ref{thm41} again, we conclude that there exist $\theta _{n,j}^1,\theta _{n,j}^2\in\sqparen{0,\pi}$ such that the eigenvalues of $C_j$ are $1$, $e^{i2\theta _{n,j}^1}$, $e^{i2\theta _{n,j}^2}$, $e^{i2(\theta _{n,j}^1+\theta _{n,j}^2)}$ with corresponding eigenvectors $u_1\otimes u_1$,
$u_2\otimes u_1$, $u_1\otimes u_2$, $u_2\otimes u_2$.

The Hilbert space $H_{n+1}$ can be decomposed into the direct sum

\begin{equation}        
\label{eq42}
H_{n+1}=H_{n+1,0}\oplus H_{n+1,1}\oplus\cdots\oplus H_{n+1,4^{n-1}-1}
\end{equation}
where $\dim H_{n+1,j}=4$ and an orthonormal basis for $H_{n+1,j}$ is

\begin{equation}        
\label{eq43}
\brac{\xhat _{n+1,4j+k}\colon k=0,1,2,3}
\end{equation}
We now define a stochastic unitary operator $V_{n+1,j}\colon H_{n+1,j}\to H_{n+1,j}$ with matrix representation $C_j$ given by \eqref{eq41}. To be explicit, we have

\begin{equation*}
\begin{bmatrix}V_{n+1,j}(\xhat _{n+1,4j})\\V_{n+1,j}(\xhat _{n+1,4j+1})\\V_{n+1,j}(\xhat _{n+1,4j+2})\\V_{n+1,j}(\xhat _{n+1,4j+3})\end{bmatrix}
=C_j\begin{bmatrix}\xhat _{n+1,4j}\\\xhat _{n+1,4j+1}\\\xhat _{n+1,4j+2}\\\xhat _{n+1,4j+3}\end{bmatrix}
\end{equation*}
We conclude from our previous work that the eigenvalues of $V_{n+1,j}$ are 
$1$, $e^{i2\theta _{n,j}^1}$, $e^{i2\theta _{n,j}^2}$, $e^{i2(\theta _{n,j}^1+\theta _{n,j}^2)}$ with corresponding unit eigenvectors

\begin{align}        
\label{eq44}
\xhat _{n+1,j}^0&=\tfrac{1}{2}(\xhat _{n+1,4j}+\xhat _{n+1,4j+1}+\xhat _{n+1,4j+2}+\xhat _{n+1,4j+3})\notag\\
\xhat _{n+1,j}^1&=\tfrac{1}{2}(\xhat _{n+1,4j}-\xhat _{n+1,4j+1}+\xhat _{n+1,4j+2}-\xhat _{n+1,4j+3})\notag\\
\xhat _{n+1,j}^2&=\tfrac{1}{2}(\xhat _{n+1,4j}+\xhat _{n+1,4j+1}-\xhat _{n+1,4j+2}-\xhat _{n+1,4j+3})\\
\xhat _{n+1,j}^3&=\tfrac{1}{2}(\xhat _{n+1,4j}-\xhat _{n+1,4j+1}-\xhat _{n+1,4j+2}+\xhat _{n+1,4j+3})\notag
\end{align}
Finally, we define the stochastic unitary operator $V_{n+1}$ on $H_{n+1}$ by

\begin{equation*}
V_{n+1}=V_{n+1,0}\oplus V_{n+1,1}\oplus\cdots\oplus V_{n+1,4^{n-1}-1}
\end{equation*}
The eigenvalues of $V_{n+1}$ are $1$ (with multiplicity $4^{n-1}$), $e^{i2\theta _{n+1,j}^1}$, $e^{i2\theta _{n+1,j}^2}$,
$e^{i2(\theta _{n+1,j}^1+\theta _{n+1,j}^2)}$, $j=0,1,\ldots ,4^{n-1}-1$. The corresponding eigenvectors are $\xhat _{n+1,j}^k$,
$j=0,1,\ldots ,4^{n-1}-1$, $k=0,1,2,3$ given by \eqref{eq44}.

The operator $V_{n+1}$ provides an intrinsic symmetry on $H_{n+1}$ generated by the coupling constants. We call $V_{n+1}$ the
\textit{coupling constant symmetry}. Since $V_{n+1}$ is unitary, it has the form $V_{n+1}=e^{iK_{n+1}}$ where $K_{n+1}$ is a self-adjoint operator called the \textit{coupling energy}. The eigenvectors of $K_{n+1}$ are again $\xhat _{n+1,j}^k$ and the corresponding eigenvalues are
$0$ (with multiplicity $4^{n-1}$), $2\theta _{n+1,j}^1$, $2\theta _{n+1,j}^2$, $2(\theta _{n+1,j}^1+\theta _{n+1,j}^2)$. Since the eigenvalues of
$K_{n+1}$ correspond to energy values, we conclude that the physical significance of the $\theta _{n+1,j}^k$ (and $0$) are that they are one-half of energy values.

\section{Position and Momentum Observables} 

It is natural to define the \textit{position observable} $Q_n$ on $H_n$ by $Q_n\xhat _{n,j}=j\xhat _{n,j}$. Then $Q_n$ is a self-adjoint operator with eigenvectors $\xhat _{n,j}$ and corresponding eigenvalues $j$, $j=0,1,\ldots ,4^{n-1}-1$. It is also natural to define the
\textit{conjugate momentum observable} $P_n$ on $H_n$ by $P_n=V_nQ_nV_n^*$. Then $P_n$ is a self-adjoint operator with the same eigenvalues and corresponding eigenvectors $V_n\xhat _{n,j}$. Now $Q_n$ describes the total position of a causet $x_{n,j}\in\pscript _n$ but it is important to describe the coordinate observables $Q_n^k$, according to the \textit{directions} $k=0,1,2,3$.

To accomplish this, it is useful to write $j$ in its quartic representation

\begin{equation}        
\label{eq51}
j=j_{n-2}j_{n-3}\cdots j_1j_0,\quad j_i\in\brac{0,1,2,3}
\end{equation}
This representation describes the directions that a path turns when moving from $x_{1,0}$ to $x_{n,j}$. For example \eqref{eq51} represents the path that turns in direction $j_{n-2}$ at $x_{1,0}$, then turns in direction $j_{n-3},\ldots$, and finally turns in direction $j_0$ just before arriving at $x_{n,j}$. If $j$ has the form \eqref{eq51} and $k\in\brac{0,1,2,3}$ define

\begin{equation*}
j^k=j'_{n-2}j'_{n-3}\cdots j'_1j'_0
\end{equation*}
where $j'_i=1$ if $j_i=k$ and $j'_i=0$ if $j_i\ne k$. Thus, $j^k$ marks the places at which $j$ turns in direction $k$. Notice that

\begin{equation}        
\label{eq52}
j=\sum _{k=0}^3kj^k
\end{equation}
For $k=0,1,2,3$, we define the \textit{coordinate observables} $Q_n^k$ by

\begin{equation*}
Q_n^k\xhat _{n,j}=j\xhat _{n,j}
\end{equation*}
We see that $Q_n^k$ is a self-adjoint operator whose eigenvalues are the $2^{n-1}$ possible values $j_{n-1}j_{n-2}\cdots j_1j_0$ where
$j_i\in\brac{0,1}$. It follows from \eqref{eq52} that

\begin{equation}        
\label{eq53}
Q_n=\sum _{k=0}^3kQ_n^k
\end{equation}

\begin{exam}
In $H_6$, let $j=403$ whose quartic representation becomes

\begin{equation*}
j=403=256+128+16+3=1\times 4^4+2\times 4^3+1\times 4^2+0\times 4+3=12103
\end{equation*}
We then have $j^0=00010=4$, $j^1=10100=272$, $j^2=01000=64$, $j^3=1$. We then have

\begin{equation*}
0\ctimes j^0+1\ctimes j^1+2\ctimes j^2+3\ctimes j^3=272+128+3=403
\end{equation*}
The coordinate observables at $j=403$ become

\begin{equation*}
Q_6^0(\xhat _{6,j})=4\xhat _{6,j},\ Q_6^1(\xhat _{6,j})=272\xhat _{6,j},\  Q_6^2(\xhat _{6,j})=64\xhat _{6,j},\ Q_6^3(\xhat _{6,j})=\xhat _{6,j}
\end{equation*}
\end{exam}

As before, we define the \textit{conjugate momentum observables} $P_n^k$ on $H_n$ by $P_n^k=V_nQ_n^kV_n^*$, $k=0,1,2,3$. To obtain explicit forms for these operators it is convenient to consider their action on $H_{n+1,j}$ given by \eqref{eq42} and having basis \eqref{eq43}. As in Section~4 we write

\begin{align*}
Q_{n+1}^k=Q_{n+1,0}^k\oplus Q_{n+1,1}^k\oplus\cdots\oplus Q_{n+1,4^{n-1}-1}^k\\
\intertext{and}
P_{n+1}^k=P_{n+1,0}^k\oplus P_{n+1,1}^k\oplus\cdots\oplus P_{n+1,4^{n-1}-1}^k
\end{align*}
Now $x_{n,j}\to x_{n+1,4j}, x_{n+1,4j+1}, x_{n+1,4j+2}, x_{n+1,4j3}$ and in quartic notation we have $j=j_{n-2}j_{n-3}\cdots j_1j_0$ while for $k=0,1,2,3$ we have

\begin{equation*}
4j+k=j_{n-2}j_{n-3}\cdots j_1j_0k
\end{equation*}
It follows that

\begin{equation*}
(4j+\ell )^k=\begin{cases}j^k1&\hbox{if }\ell=k\\j^k0&\hbox{if }\ell\ne k\end{cases}
\end{equation*}
Hence,

\begin{equation*}
Q_{n+1}^k\xhat _{n+1,4j+\ell}=(4j+\ell )^k\xhat _{n+1,4j+\ell}
=\begin{cases}j^k1\xhat _{n+1,4j+\ell}&\hbox{if }\ell=k\\j^k0\xhat _{n+1,4j+\ell}&\hbox{if }\ell\ne k\end{cases}
\end{equation*}
Thus, in the standard basis $\brac{\xhat _{n+1,4j+k}\colon k=0,1,2,3}$ for $H_{n+1,j}$, the operator $Q_{n+1,j}^0$ has matrix form

\begin{align*}
Q_{n+1,j}^0&=
\begin{bmatrix}j^01&0&0&0\\0&j^00&0&0\\0&0&j^00&0\\0&0&0&j^00\end{bmatrix}
=(j^00)I+\begin{bmatrix}1&0&0&0\\0&0&0&0\\0&0&0&0\\0&0&0&0\end{bmatrix}\\
&=(j^00)I+\ket{\xhat _{n+1,4j}}\bra{\xhat _{n+1,4j}}
\end{align*}
Similarly, we have for $k=1,2,3$ that

\begin{equation}        
\label{eq54}
Q_{n+1,j}^k=(j^k0)I+\ket{\xhat _{n+1,4j+k}}\bra{\xhat _{n+1,4j+k}}
\end{equation}
It follows that the operators $P_{n+1,j}^k$ have the form

\begin{align}        
\label{eq55}
P_{n+1,j}^k&=V_{n+1,j}Q_{n+1,j}^kV_{n+1,j}^*\notag\\
&=(j^k0)I+\ket{V_{n+1}\xhat _{n+1,4j+k}}\bra{V_{n+1}\xhat _{n+1,4j+k}}
\end{align}
for $k=0,1,2,3$

As usual in quantum theory, it is of interest to find the commutator of $Q_{n+1,j}^k$ and $P_{n+1,j}^\ell$. This becomes

\begin{align*}
\sqbrac{Q_{n+1,j}^k,P_{n+1,j}^\ell}&=Q_{n+1,j}^kP_{n+1,j}^\ell-P_{n+1,j}^\ell Q_{n+1,j}^k\\
  &=\ket{\xhat _{n+1,4j+k}}\bra{\xhat _{n+1,4j+k}}\ket{V_{n+1}\xhat _{n+1,4j+\ell}}\bra{V_{n+1}\xhat _{n+1,4j+\ell}}\\
  &\quad -\ket{V_{n+1}\xhat _{n+1,4j+\ell}}\bra{V_{n+1}\xhat _{n+1,4j+\ell}}\ket{\xhat _{n+1,4j+k}}\bra{\xhat _{n+1,4j+k}}
\end{align*}
Letting

\begin{equation*}
R_{n+1,j}^{k,\ell}=\elbows{\xhat _{n+1,4j+k},V_{n+1}\xhat _{n+1,4j+\ell}}\ket{\xhat _{n+1,4j+k}}\bra{V_{n+1}\xhat _{n+1,4j+\ell}}
\end{equation*}
we have that

\begin{equation}        
\label{eq56}
\sqbrac{Q_{n+1,j}^k,P_{n+1,j}^\ell}=R_{n+1,j}^{k,\ell}-(R_{n+1,j}^{k,\ell})^*=\tfrac{1}{2i}\,\rmim (R_{n+1,j}^{k,\ell})
\end{equation}
We can apply \eqref{eq41} and \eqref{eq56} to find explicit expressions for the commutators. For example

\begin{align*}
R_{m+1.k}^{0,0}&=\elbows{\xhat _{n+1,4j},V_{n+1}\xhat _{n+1,4j}}\ket{\xhat _{n+1,4j}}\bra{V_{n+1}\xhat _{n+1,4j}}\\
  &=c_{n,j}^0\sqbrac{\ket{\xhat _{n+1,4j}}\bra{\sum _{k=0}^3c_{n,j}^k\xhat _{n+1,4j+k}}}\\\noalign{\smallskip}
  &=c_{n,j}^0\sqbrac{\sum _{k=0}^3\cbar _{n,j}^k\ket{\xhat _{n+1,4j}}\bra{\xhat _{n+1,4j+k}}}
\end{align*}
The matrix representation for $R_{n+1,j}^{0,0}$ is

\begin{equation*}
R_{n+1,j}^{0,0}=c_{n,j}^0
\begin{bmatrix}\cbar_{n,j}^0&\cbar_{n,j}^1&\cbar_{n,j}^2&\cbar_{n,j}^3\\0&0&0&0\\0&0&0&0\\0&0&0&0\end{bmatrix}
\end{equation*}
We conclude that

\begin{align}        
\label{eq57}
\sqbrac{Q_{n+1,j}^0,P_{n+1,j}^0}&=R_{n+1,j}^{0,0}-(R_{n+1,j}^{0,0})^*\notag\\\noalign{\smallskip}
&=\begin{bmatrix}0&c_{n,j}^0\cbar_{n,j}^1&c_{n,j}^0\cbar_{n,j}^2&c_{n,j}^0\cbar_{n,j}^3\\\noalign{\smallskip}
-\cbar _{n,j}^0c_{n,j}^1&0&0&0\\\noalign{\smallskip}-\cbar _{n,j}^0c_{n,j}^2&0&0&0\\\noalign{\smallskip}
-\cbar _{n,j}^0c_{n,j}^3&0&0&0\end{bmatrix}
\end{align}
As another example we have

\begin{align*}
\sqbrac{Q_{n+1,j}^1,P_{n+1,j}^0}&=R_{n+1,j}^{1,0}-(R_{n+1,j}^{1,0})^*\notag\\\noalign{\smallskip}
&=\begin{bmatrix}0&-\cbar_{n,j}^1c_{n,j}^0&0&0\\\noalign{\smallskip}
  c_{n,j}^1\cbar _{n,j}^0&0&c_{n,j}^1\cbar _{n,j}^2&c_{n,j}^1\cbar _{n,j}^3\\\noalign{\smallskip}
  0&-\cbar _{n,j}^1c_{n,j}^2&0&0\\\noalign{\smallskip}0&-\cbar _{n,j}^1c_{n,j}^3&0&0\end{bmatrix}
\end{align*}
The other commutators are similar and the full operators are given by

\begin{equation*}
\sqbrac{Q_{n+1}^k,P_{n+1}^\ell}=\sqbrac{Q_{n+1,0}^k,P_{n+1,0}^\ell}\oplus\cdots\oplus\sqbrac{Q_{n+1,4^{n-1}-1}^k,P_{n+1,4^{n-1}-1}^\ell}
\end{equation*}
Any one of these commutators $A$ is anti-self-adjoint in the sense that $A^*=-A$. Thus they have a complete set of eigenvectors and corresponding eigenvalues. The next result considers the case of $\sqbrac{Q_{n+1,j}^0,P_{n+1,j}^0}$ and the others are similar. We assume that $c_{n,j}^0\ne 0$ because otherwise the result is trivial.

\begin{thm}       
\label{thm51}
The eigenvalues of $\sqbrac{Q_{n+1,j}^0,P_{n+1,j}^0}$ are $0$ (with multiplicity 2) and\newline
$\lambda _\pm =\pm i\ab{c_{n,j}^0}\sqbrac{1-\ab{c_{n,j}^0}^2}^{1/2}$ with corresponding eigenvectors

\begin{equation*}
v_r=(0,a_r,b_r,c_r),\quad r=1,2
\end{equation*}
where $(a_r,b_r,c_r)\perp (\cbar _{n,j}^1,\cbar _{n,j}^2,\cbar _{n,j}^3)$, $r=1,2$ with $v_1\perp v_2$, $v_r\ne 0$, and

\begin{equation*}
v_\pm =(\lambda _\pm ,-\cbar _{n,j}^0c_{n,j}^1,-\cbar _{n,j}^0c_{n,j}^2,-\cbar _{n,j}^0c_{n,j}^3)
\end{equation*}
\end{thm}
\begin{proof}
From its form in \eqref{eq57} it is clear that $\sqbrac{Q_{n+1,j}^0,P_{n+1,j}^0}v_r=0$, $=1,2$. Moreover, from \eqref{eq57} we have

\begin{align*}
\sqbrac{Q_{n+1,j}^1,P_{n+1,j}^0}v_\pm&=
 \begin{bmatrix}-\ab{c_{n,j}^0}^2(1-\ab{c_{n,j}^0}^2)\\\noalign{\smallskip}
  \lambda _\pm (-\cbar _{n,j}^0c_{n,j}^1)\\\noalign{\smallskip}
  \lambda _\pm (-\cbar _{n,j}^0c_{n,j}^2)\\\noalign{\smallskip}
  \lambda _\pm (-\cbar _{n,j}^0c_{n,j}^3)\end{bmatrix}\\\noalign{\smallskip}
   &=\lambda _\pm\begin{bmatrix}-\ab{c_{n,j}^0}^2(1-\ab{c_{n,j}^0}^2)/\lambda _\pm\\\noalign{\smallskip}
 -\cbar _{n,j}^0c_{n,j}^1\\\noalign{\smallskip} -\cbar _{n,j}^0c_{n,j}^2\\\noalign{\smallskip}
 -\cbar _{n,j}^0c_{n,j}^3\end{bmatrix}\\
 &=\lambda _\pm v_\pm\qedhere
\end{align*}
\end{proof}
Of course, we can now apply Theorem~\ref{thm51} to find the eigenvalues and eigenvectors of $\sqbrac{Q_{n+1}^0,P_{n+1}^0}$.

\section{Expectations} 

In this section we compute expectations of observables considered in Section~5. These expectations will be computed relative to the amplitude state
$\ahat _n$ of Equation~\eqref{eq34}. We denote the expectation of an operator $A$ in the state $\ahat _n$ by

\begin{equation*}
E_n(A)=\elbows{\ahat _n,A\ahat _n}
\end{equation*}
For the position observable $Q_n$, we have that

\begin{align*}
E_n(Q_n)&=\elbows{\ahat _n,Q_n\ahat _n}=\sum _{r=0}^{4^{n-1}-1}\sum _{s=0}^{4^{n-1}-1}\abar (x_{n,r})sa(x_{n,s})\delta _{r,s}\\
  &=\sum _{j=0}^{4^{n-1}-1}j\ab{a(x_{n,j})}^2
\end{align*}
which is not at all surprising. In a similar way, we obtain

\begin{equation*}
E_n(Q_n^k)=\sum _{j=0}^{4^{n-1}-1}j^k\ab{a(x_{n,j})}^2
\end{equation*}

To describe momentum operators, it is again convenient to consider $H_{n+1}$ with the amplitude state

\begin{equation*}
\ahat _{n+1}=\sum _{k=0}^3\sum _{j=0}^{4^{n-1}-1}a(x_{n+1,4j+k})\xhat _{n+1,4j+k}
\end{equation*}
As before, we first consider $P_{n+1,j}^k$ on $H_{n+1,j}$ and define

\begin{equation*}
\ahat _{n+1,j}=\sum _{k=0}^3a(x_{n+1,4j+k})\xhat _{n+1,4j+k}
\end{equation*}
By \eqref{eq55} we have that

\begin{align*}
E_{n+1,j}(P_{n+1,j}^k)&=\elbows{\ahat _{n+1,j},P_{n+1,j}\ahat _{n+1,j}}\\
  &=(j^k0)\|\ahat _{n+1,j}\|^2+\ab{\elbows{V_{n+1}\xhat _{n+1,4j+k},\ahat _{n+1,j}}}^2
\end{align*}

In particular,

\begin{align*}
E_{n+1,j}(P_{n+1,j}^0&=(j^00)\|\ahat _{n+1,j}\|^2+\ab{\elbows{V_{n+1}\xhat _{n+1,4j},\ahat _{n+1,j}}}^2\\
  &=(j^00)\|\ahat _{n+1,j}\|^2+\ab{\elbows{\sum _{k=0}^3c_{n,j}^k\xhat _{n+1,4j+k},\ahat _{n+1,j}}}^2\\
  &=(j^00)\|\ahat _{n+1,j}\|^2+\ab{\sum _{k=0}^3\cbar _{n,j}^ka(x_{n+1,4j+k})}^2\\
  &=(j^00)\|\ahat _{n+1,j}\|^2+\ab{\sum _{k+0}^3\cbar _{n,j}^kc_{n,j}^ka(x_{n,j})}^2\\
  &=(j^00)\|\ahat _{n+1,}\|^2+\ab{a(x_{n,j})}^2\\
  &=\sqbrac{(j^00)+1}\ab{a(x_{n,j})}^2
\end{align*}
We conclude that

\begin{equation*}
E_{n+1}(P_{n+1}^0)=\sum _{j=0}^{4^{n-1}-1}(j^00)\ab{a(x_{n,j})}^2+1
\end{equation*}

Surprisingly, we obtain a slightly different expression for $E_{n+1,j}(P_{n+1}^1)$. In particular,

\begin{equation*}
E_{n+1,j}(P_{n+1,j}^1)=(j^10)\|\ahat _{n+1,j}\|^2+\ab{\elbows{V_{n+1}\xhat _{n+1,4j+1},\ahat _{n+1,j}}}^2
\end{equation*}
But

\begin{equation*}
V_{n+1}\xhat _{n+1,4j+1}=c_{n,j}^1\xhat _{n+1,4j}+c_{n,j}^0\xhat _{n+1,4j+1}+c_{n,j}^3\xhat _{n+1,4j+2}+c^2_{n,j}x_{n+1,4j+3}
\end{equation*}
It follows that

\begin{align*}
\elbows{V_{n+1}\xhat _{n+1,4j+1},\ahat _{n+1,j}}&=\cbar_{n,j}^1a(x_{n+1,4j})+\cbar _{n,j}^0a(x_{n+1,4j+1})+\cbar _{n,j}^3a(x_{n+1,4j+2})\\
  &\quad +\cbar _{n,j}^2a(x_{n+1,4j+3})\\
  &=(\cbar _{n,j}^1c_{n,j}^0+\cbar _{n,j}^0c_{n,j}^1+\cbar _{n,j}^3c_{n,j}^2+\cbar _{n,j}^2c_{n,j}^3)a(x_{n,j})\\
  &=0
\end{align*}
Hence,

\begin{equation*}
E_{n+1,j}(P_{n+1,j}^1)=(j^10)\ab{a(x_{n,j})}^2
\end{equation*}
so that

\begin{equation*}
E_{n+1}(P_{n+1}^1)=\sum _{j=0}^{4^{n-1}-1}(j^10)\ab{a(x_{n,j})}^2
\end{equation*}
In a similar way we have

\begin{equation*}
E_{n+1}(P_{n+1}^k)=\sum _{j=0}^{4^{n-1}-1}(j^k0)\ab{a(x_{n,j})}^2
\end{equation*}
for $k=2,3$.

Even though $\sqbrac{Q_{n+1}^k,P_{n+1}^\ell}$ is not self-adjoint, the operator $i\sqbrac{Q_{n+1}^k,P_{n+1}^\ell}$ is self-adoint and it is of interest to find its expectation. We begin with

\begin{equation*}
E_{n+1,j}\brac{\sqbrac{Q_{n+1,j}^0,P_{n+1,j}^0}}=\elbows{\ahat _{n+1,j},\sqbrac{Q_{n+1,j}^0,P_{n+1,j}^0}\ahat _{n+1,j}}
\end{equation*}
Applying \eqref{eq57} we have that

\begin{align*}
\sqbrac{Q_{n+1,j}^0,P_{n+1,j}^0}\ahat _{n+1,j}&=
 \begin{bmatrix}c_{n,j}^0\sum _{k=1}^3\cbar _{n,j}^ka(x_{n+1,4j+k})\\\noalign{\smallskip}
  -\cbar _{n,j}^0c_{n,j}^1a(x_{n+1,4j})\\\noalign{\smallskip}
 -\cbar _{n,j}^0c_{n,j}^2a(x_{n+1,4j})\\\noalign{\smallskip}
  -\cbar _{n,j}^0c_{n,j}^3a(x_{n+1,4j})\end{bmatrix}\\\noalign{\smallskip}
   &=a(x_{n,j})\begin{bmatrix}c_{n,j}^0\sum _{k=1}^3\ab{c_{n,j}^k}^2\\\noalign{\smallskip}
 -\ab{c_{n,j}^0}^2c_{n,j}^1\\\noalign{\smallskip} -\ab{c_{n,j}^0}^2c_{n,j}^2\\\noalign{\smallskip}
 -\ab{c_{n,j}^0}^2c_{n,j}^3\end{bmatrix}\\
\end{align*}
We conclude that

\begin{align*}
E_{n+1,j}&\brac{\sqbrac{Q_{n+1,j}^0P_{n+1,j}^0}}\\
  &=\ab{a(x_{n,j})}^2\sqbrac{\ab{c_{n,j}^0}^2\sum _{k=1}^3\ab{c_{n,j}^k}^2-\ab{c_{n,j}^0}^2\ab{c_{n,j}^1}^2
  -\ab{c_{n,j}^0}^2\ab{c_{n,j}^2}^2-\ab{c_{n,j}^0}^2\ab{c_{n,j}^3}^2}\\
  &=0
\end{align*}
It follows that $E_{n+1}\brac{\sqbrac{Q_{n+1}^0,P_{n+1}^0}}=0$. Hence, for the amplitude state there is no lower bound for the product of the variances of $Q_{n+1}^0$ and $P_{n+1}^0$ as in the Heisenberg uncertainty relation. In a similar way, one can show that

\begin{equation*}
E_{n+1}\brac{\sqbrac{Q_{n+1}^k,P_{n+1}^\ell}}=0
\end{equation*}
for $k,\ell =0,1,2,3$.

\bibliographystyle{mdpi}
\makeatletter
\renewcommand\@biblabel[1]{#1. }
\makeatother



%


%

\end{document}